\documentclass[10pt,a4paper]{article}
\usepackage[margin=2cm]{geometry}

\usepackage{enumerate}
\usepackage[utf8]{inputenc}
\usepackage{amsmath,amssymb,tabu}
\usepackage{subdepth}
\usepackage{color}
\usepackage{calc}
\usepackage{tikz}
\usetikzlibrary{positioning,arrows,decorations.pathreplacing,shapes}
\usetikzlibrary{calc}
\usepackage{multirow}
\usepackage{boxedminipage}
\usepackage{xifthen}
\usepackage{tabularx}
\usepackage{xcolor}
\usepackage{setspace}
\usepackage{authblk}
\usepackage{thmtools, thm-restate}

\usepackage{mathpazo}

\usepackage{caption, subcaption}

\usepackage{algorithm}
\usepackage{algpseudocode}
\usepackage{hyperref}
\usepackage{cleveref}

\newlength{\bibitemsep}
\newlength{\bibparskip}
\let\oldthebibliography\thebibliography
\renewcommand\thebibliography[1]{%
	\oldthebibliography{#1}%
	\setlength{\parskip}{\bibitemsep}%
	\setlength{\itemsep}{\bibparskip}%
}

\newcommand{\R}{\mathbb{R}}

\newcommand{\M}{\mathcal{M}}
\newcommand{\bM}{M}

\newcommand{\I}{\mathcal{I}}

\newcommand{\del}{\backslash}

\newcommand\bmat[1]{\begin{bmatrix} #1 \end{bmatrix}}

\newcommand{\OPT}{OPT}
\newcommand{\norm}[1]{\left\lVert#1\right\rVert}

\DeclareMathOperator{\E}{\mathbb{E}}
\DeclareMathOperator{\Prob}{\mathrm{Pr}}

\newtheorem{theorem}{Theorem}
\newtheorem{corollary}{Corollary}

\newtheorem{lemma}{Lemma}

\newenvironment{proof}[1][]{\par \noindent {\bf Proof #1}\ }{\hfill$\Box$\par \vspace{11pt}}

\usepackage[textsize=tiny,textwidth=2cm,color=green!50!gray]{todonotes} %,disable

\newcommand*\samethanks[1][\value{footnote}]{\footnotemark[#1]}
\title{Maximizing the Minimum Eigenvalue in Constant Dimension}

\author[1]{Adam Brown\thanks{ajmbrown@gatech.edu, msingh94@gatech.edu; supported in part by NSF CCF-2106444 and NSF CCF-1910423.}}
\author[2]{Aditi Laddha \thanks{aditi.laddha@yale.edu; supported in part by the Institute for Foundations of Data Science at Yale and NSF CCF- 2007443.}}
\author[1]{Mohit Singh\samethanks[1]}

\affil[1]{Georgia Institute of Technology.}
\affil[2]{Yale University.}

\begin{document}
\clearpage
\maketitle
\thispagestyle{empty}
\begin{abstract}
   In an instance of the minimum eigenvalue problem, we are given a collection of $n$ vectors $v_1,\ldots, v_n \subset {\R^d}$, and the goal is to pick a subset $B\subseteq [n]$ of given vectors to maximize the minimum eigenvalue of the matrix $\sum_{i\in B} v_i v_i^{\top} $. Often, additional combinatorial constraints such as cardinality constraint $\left(|B|\leq k\right)$ or matroid constraint ($B$ is a basis of a matroid defined on $[n]$) must be satisfied by the chosen set of vectors. The minimum eigenvalue problem with matroid constraints  models a wide variety of problems including the Santa Clause problem, the E-design problem, and the constructive Kadison-Singer problem.

    In this paper, we give a randomized algorithm that finds a set $B\subseteq [n]$ subject to any matroid constraint whose minimum eigenvalue is at least $(1-\epsilon)$ times the optimum, with high probability. The running time of the algorithm is $O\left( n^{O(d\log(d)/\epsilon^2)}\right)$. In particular, our results give a polynomial time asymptotic scheme when the dimension of the vectors is constant. Our algorithm uses a convex programming relaxation of the problem after guessing a rescaling which allows us to apply pipage rounding and matrix Chernoff inequalities to round to a good solution. The key new component is a structural lemma which enables us to ``guess'' the appropriate rescaling, which could be of independent interest. Our approach generalizes the approximation guarantee to monotone, homogeneous functions and as such we can maximize $\det(\sum_{i\in B} v_i v_i^\top)^{1/d}$, or minimize any norm of the eigenvalues of the matrix $\left(\sum_{i\in B} v_i v_i^\top\right)^{-1} $, with the same running time under some mild assumptions. As a byproduct, we also get a simple algorithm for an algorithmic version of Kadison-Singer problem.
\end{abstract}
%\nonumber
\newpage
\clearpage
\setcounter{page}{1}
%%%%%%%%%%%%%%%%%%%%%%%%%%%%%%%%%%%%%%%%%%%%%%%%%%%%%%%%%%%%
\section{Introduction}

Subset selection problems with spectral objectives offer a natural model for studying problems in a variety of fields, including numerical linear algebra~\cite{AvronBoutsidis2013}, graph theory~\cite{Batson2008}, convex geometry~\cite{DiSumma2015, J-simplex}, resource allocation~\cite{AsadpourSaberiMaxMin,ChuzhoyMaxMin}, and optimal design of experiments~\cite{Allen-Zhu2017,a-design-hard,LauZhouLocalSearch}, all under a single umbrella. 

In this work, we consider the minimum eigenvalue problem. In an instance of a minimum eigenvalue problem, we are given a collection of $n$ vectors $v_1,\ldots, v_n \in {\R^d}$, and the goal is to pick a subset $B\subseteq [n]$ of given vectors to maximize the minimum eigenvalue of the matrix $\sum_{i\in B} v_i v_i^\top $. The selected set $B$ must satisfy additional constraints such as cardinality, partition, or more generally matroid constraints. While much of the focus in previous works~\cite{Allen-Zhu2017,a-design-hard,LauZhouLocalSearch} has been on cardinality constraints, in this work, we consider general matroid constraints. The generality of matroid constraints allows us to model the algorithmic version of the Kadison Singer problem~\cite{MarcusSpielmanSrivastavaKS,jourdan2022algorithmic} as well as the Santa Claus allocation problem~\cite{AsadpourSaberiMaxMin,AsadpourFeigeSaberiHypergraph,FeigeFairness,ChuzhoyMaxMin,DaviesRothvossZhangMatroidSanta} as a special case of the minimum eigenvalue problem.

The ``discrepancy'' formulation of the Kadison-Singer problem (shown to be equivalent to the original formulation in~\cite{weaver2004kadison} and proved in~\cite{MarcusSpielmanSrivastavaKS}) states that given a set of vectors $v_1, \ldots, v_m \in \R^d$ with $\norm{v_i} \leq \alpha $ and $\sum_i v_iv_i^\top = I_d$, there exists a partition of $[m]$ into two subsets $S_1, S_2$, such that for every $j$, $\sum_{i \in S_j} v_i v_i^\top$ spectrally approximates $I_d/2$ to an additive factor of $O(\alpha)$. Algorithmically, finding such a partition is equivalent to solving an instance of the minimum eigenvalue maximization problem under partition matroid constraints (see Section~\ref{sec:ks}).

Another classical application of the minimum eigenvalue problem arises in the area of optimal design of experiments in statistics~\cite{pukelsheim2006optimal,Allen-Zhu2017}. The goal in the design of experiments is to select a subset of vectors $S$ from a given list of vectors $\{v_1,\ldots, v_n\}$ such that certain measures of the covariance matrix $\left(\sum_{i\in S} v_i v_i^{\top}\right)^{-1}$ are small. In particular, minimizing the maximum eigenvalue of the covariance matrix, classically known as the $E$-design problem in statistics, is exactly the minimum eigenvalue problem. While much of the previous work has focused on the case when the selected set of measurements $S$ must satisfy cardinality constraints, our work generalizes this problem to be studied under general matroid constraints. 

\subsection{Our Results and Contributions}
In this work, we present an approximation algorithm for the minimum eigenvalue problem for all matroids. We use the randomized rounding technique of pipage rounding to give a  polynomial time approximation scheme (PTAS) when the dimension is constant. 
\begin{theorem} \label{thm:min_lambda}
    For any $\epsilon > 0$ there is an $O\left( n^{O(d\log(d)/\epsilon^2)}\right)$-time algorithm which, given a collection of vectors $v_1,\ldots, v_n \in \R^d$ and a matroid $\M = ([n], \I)$ returns a set $B \in \I$ such that with probability at least $1-d^{-4}$
    \[ \lambda_{\min}\left(\sum_{i\in B} v_iv_i^\top\right) \geq (1-\epsilon) \cdot \max_{B^\star \in \I} \; \lambda_{\min}\left(\sum_{i\in B^\star} v_iv_i^\top\right).\]
\end{theorem}

Our result generalizes to give a PTAS (for constant dimension) when the objective is a general matrix function satisfying certain technical properties. In particular, this implies that a similar result as in Theorem~\ref{thm:min_lambda} is achievable when the objective is to maximize the determinant of $\sum_{i\in B} v_i v_i^\top$ or to minimize any norm of the eigenvalues of $(\sum_{i\in B} v_i v_i^\top)^{-1}$. %\MS{Add other objectives here.}
\begin{theorem}\label{thm:general_f}
Suppose we have a collection of vectors $\mathcal{V} = (v_1,\ldots, v_n) \in \R^d$, and a matroid $\M = ([n], \I)$. Let $f: \mathbb{S}_{d}^+: \rightarrow \mathbb{R}$ be a concave, monotone, and homogeneous function given with a value and first order oracle.
%maximizes the expression $f(x_1 v_1 v_1^\top + x_2 v_2 v_2^\top + \ldots + x_n v_n v_n^\top)$ in polynomial time.
    For any $\epsilon > 0$, there is an $O\left( n^{O(d\log(d)/\epsilon^2)}\right)$-time randomized algorithm, which takes $(\mathcal{V}, \mathcal{M}, f)$ as input and returns a set $B \in \I$ such that with probability at least $1-d^{-4}$,
    \begin{equation*}
        f\left(\sum_{i\in B} v_iv_i^\top\right) \geq (1-\epsilon) \cdot \max_{B^\star \in \I} \; f\left(\sum_{i\in B^\star} v_iv_i^\top\right) \,.
    \end{equation*}
\end{theorem}
Although Theorem~\ref{thm:general_f} is stated in terms of maximizing concave functions, our algorithm can also be applied to minimize monotone and homogeneous convex functions (e.g., $\mathrm{trace}(\sum_{i\in B} v_i v_i^\top)^{-1}$) by considering the natural convex relaxation of the function over the matroid base polytope and using the same rounding strategy. 
\paragraph{Technical Overview.} The first natural direction is to construct a convex programming relaxation for the problem and aim to apply randomized rounding methods to it. 
\begin{equation*} \label{eq:cvx_natural} \tag{CP} 
    \begin{array}{cl}
        \max & \lambda_{\min}\left( X\right) \\
        & X = \sum\limits_{i=1}^n x_{i} \cdot v_{i}v_{i}^\top \\
        & x \in \mathcal{P}(\M) \\
        & x \geq 0
    \end{array}
\end{equation*}
Here, $\mathcal{P}(\M) $ denotes the matroid base polytope of $\mathcal{M}$.
 Unfortunately, this direct approach faces problems as this natural relaxation has an unbounded integrality gap even in very special cases (see Appendix~\ref{sec:integrality_gap}). The main challenge is the presence of \emph{long} vectors that contribute significantly towards the optimum solution. A natural way to formalize the contribution of a vector is to consider its \emph{leverage score}. Indeed, if $T$ denotes the optimum solution and $A_T=\sum_{i\in T}v_i v_i^{\top}$, let $l_i=v_i^{\top} A_T^{-1} v_i $  be the leverage score of $v_i$ and let $S=\{i\in T: l_i\geq \frac{\epsilon^2}{\log d}\}$ be the set of vectors in the optimum solution with large leverage scores.  The bound  $\frac{\epsilon^2}{\log d}$ is chosen to allow randomized rounding methods to work (see Lemma~\ref{cor:flexible-matrix-chernoff} for details). Using the simple fact that the sum of leverage scores of all vectors in the optimum solution is exactly $d$, it follows that $|S|\leq \frac{d \log d}{\epsilon^2}$. Thus we could easily enumerate all such subsets $S$ in time $n^{O(d \log d/\epsilon^2)}$. For each such guess $S$, we include $S$ in our solution and solve the convex program. We then apply the randomized rounding method to the solution of the convex program. Unfortunately, the challenge lies in ensuring that the convex program not only selects the vectors in $S$ (this can be easily done by setting their indicator variable to one) but also avoids selecting all vectors \emph{not} in $T$ that have a large leverage score. The latter is crucial for the randomized rounding approach to work effectively. Unfortunately, since we did not guess $T$, we have no way to insist that we do not pick these vectors in the convex program. 

To address this problem, we present a new structural lemma that enables us to compute the leverage score as given by matrix $A_S=\sum_{i\in S} v_i v_i^T$. The lemma shows that there are few vectors with large leverage scores, even when using $A_S$ instead of $A_T$. Observe that $A_S^{-1}\succeq A_T^{-1}$ and therefore, the leverage scores with respect to $A_S$ are larger. Nevertheless, we still show a similar bound in the following lemma. 

\begin{restatable}[]{lemma}{localsearch}\label{lem:local_search}
    For any set $T$ and a set of vectors $\{v_i:i\in T\}$ in $\mathbb{R}^d$ such that $\sum_{i \in T} v_i v_i^\top$ is invertible, there exists a subset $S \subseteq T$ such that $|S| = O( d\log(d) /\epsilon^2 )$, $A_S = \sum_{i \in S} v_i v_i^\top$ is invertible, and for all $i \in T\del S$,
    \begin{equation*}
        v_i^\top A_S^{-1} v_i \leq \frac{\epsilon^2}{10\log(d)}\,.
    \end{equation*}
\end{restatable}
With the help of Lemma~\ref{lem:local_search}, we can now guess the set $S$ and insist that the convex program includes all these vectors in the chosen subset. More importantly, it allows us to insist that all vectors $v_i$ not in $S$ such that  $v_i^\top A_S^{-1} v_i > \frac{\epsilon^2}{10\log(d)}$ not be included in the chosen solution. The last step can be done since we have guessed the set $S$. This allows us to apply the randomized rounding approach to the convex programming solution.

There are some points worth mentioning about the randomized rounding approach. When the constraint matroid is a partition matroid, randomized rounding is a natural choice: for each part, the convex programming solution can be interpreted as a probability distribution over vectors in that part. Independently, for each part, pick one of the vectors with probability given by the convex programming solution. A simple application of the matrix Chernoff bound and the fact that leverage scores are all small due to Lemma~\ref{lem:local_search} gives us the desired result. Due to the simplicity of the approach for partition matroids as well as the applicability of these constraints, we first prove the result for partition matroids in Section~\ref{sec:partition}. We also show the application of our result to obtain an algorithmic version of the Kadison-Singer problem~\cite{MarcusSpielmanSrivastavaKS} for constant dimension. We slightly improve the run time compared to the recent work~\cite{jourdan2022algorithmic}.

\begin{corollary}\label{cor:ks_alg}
    Suppose we are given collection of vectors $\mathcal{U} = (u_1,\ldots, u_n) \in \R^d$ with $\|u_i\|^2 \leq \alpha$ for any $i \in [n]$ and $\sum_{i=1}^n u_iu_i^\top = I_d$ and a constant $c > 0$ such that there exists a set $T^*$ satisfying
    \begin{equation*} 
        \left(\frac{1}{2} - c\sqrt{\alpha}\right)I_d \preceq \sum_{i \in T^*} u_iu_i^\top \preceq \left(\frac{1}{2} + c\sqrt{\alpha}\right)I_d \,.
    \end{equation*}
     For any $\epsilon > 0$, there exists a randomized algorithm such which given $\mathcal{U}$ and $c$ as input, returns a set $T$ such that 
    \begin{equation*} 
        (1-\epsilon)\cdot \left(\frac{1}{2} - c\sqrt{\alpha}\right)I_d \preceq \sum_{i \in T} u_iu_i^\top \preceq (1+\epsilon)\cdot \left(\frac{1}{2} + c\sqrt{\alpha}\right)I_d \,,
    \end{equation*}
    with probability at least $1-O(d^{-4})$. The run time of the algorithm is $O(n^{O(d\log{d}/\epsilon^2)})$.
\end{corollary}

For general matroids, a straightforward application of randomized rounding does not work since it will not ensure that the chosen set is an independent set in the matroid. Instead, we use \emph{pipage rounding} for general matroids, which involves randomly walking in the matroid polytope to return a vertex while ensuring that the output solution has even better concentration than is given by independent randomized rounding. To show these concentration results, we build on the work of Harvey and Olver~\cite{harvey2014pipage} and give lower tail bounds on the distribution obtained via pipage rounding in Lemma~\ref{lem:renormalized-pipage}.

\subsection{Related Work}
The minimum eigenvalue problem with partition constraints can be interpreted as a generalization of the max-min allocation problem. In the case of cardinality constraints, it can also model problems from experimental design and spectral sparsification. We give an overview of prior work for these special cases.

\paragraph{Max-min allocation and Santa Claus:} In the max-min allocation problem, we are given a set $[d]$ of agents and a set $[n]$ of items where agent $j \in [d]$ has valuation $h_{ij}\geq 0$ for item $i$. The goal is to select an assignment $\sigma: [n] \rightarrow [d]$ which maximizes
\[ \min_{j \in [d]} \sum_{i:\sigma(i) = j} h_{ij}.\]
This can be seen as a special case of the minimum eigenvalue problem with partition constraints. 

Bansal and Sviridenko~\cite{Bansal2006} introduced the configuration LP as a relaxation for the max-min allocation problem but showed that it has an integrality gap of $\Omega(\sqrt{n})$~\cite{Bansal2006}. Asadpour and Saberi~\cite{AsadpourSaberiMaxMin} gave a rounding scheme for the same LP, which achieves an $O(\sqrt{n}\log^3 n)$-approximation. This was later improved by Chakrabarty et al.~\cite{ChuzhoyMaxMin} to an $\tilde O(n^\epsilon)$-approximation for any $\epsilon \in \Omega(\log\log n / \log n)$ by iteratively constructing new instances with smaller integrality gap.

A further special case is the Santa Claus problem where each item $i$ has an intrinsic value $H_i \geq 0$ such that $h_{ij} \in \{0,H_i\}$ for all players $j \in [d]$. Here, Bansal and Sviridenko~\cite{Bansal2006} used the configuration LP to find an $O(\log\log n /\log\log\log n)$-approximation. Feige~\cite{FeigeFairness} non-constructively showed a constant upper bound on the integrality gap of the configuration LP for the Santa Claus problem by iteratively applying the Lov{\'a}sz Local Lemma. The current best bound is due to Haxell and Szab\'o~\cite{Haxell2023}, who used new topological techniques to show that the integrality gap is at most $3.534$. Bounds on the integrality gap do not immediately lead to efficient approximation algorithms, but Davies et al.~\cite{DaviesRothvossZhangMatroidSanta} recently gave an algorithm for a more general setting that can be used to achieve a $(4+\epsilon)$-approximation for the Santa Claus problem.

\paragraph{Experimental Design (E-optimal Design):}

Even with cardinality constraints (uniform matroid of rank $k$), the minimum eigenvalue problem is NP-hard~\cite{e-design-hard}. Allen-Zhu et al.~\cite{Allen-Zhu2017} showed that it is possible to deterministically find a $(1-\epsilon)$-approximation so long as $k \geq \Omega(d/\epsilon^2)$ by rounding the natural convex relaxation. They also conjectured that this requirement was necessary. This conjecture was confirmed in~\cite{a-design-hard}, where they showed an integrality gap instance for the convex relaxation. Recently Lau and Zhou~\cite{LauZhouLocalSearch} have built on the regret minimization framework from~\cite{Allen-Zhu2017} to show that a modified local search algorithm with a ``smoothed" objective works as long as there is a near-optimal solution with a good condition number.

\paragraph{Spectral Sparsification and Kadison-Singer.} The problem of rounding the natural convex programming relaxation for the minimum eigenvalue problem is closely related to spectral sparsification~\cite{Batson2008} and the Kadison-Singer problem~\cite{MarcusSpielmanSrivastavaKS}. In spectral sparsification~\cite{Batson2008}, the goal is to pick a small subset of vectors $S\subseteq [n]$ such that $\sum_{i\in S} w_i v_i v_i^{\top}$ spectrally approximates $\sum_{i\in [n]} v_i v_i^{\top}$ for some weights $w_i$. In the cardinality constrained minimum eigenvalue problem, rounding the convex programming solution involves finding a small set $S$, such that $\sum_{i\in S} v_i v_i^{\top}$ spectrally approximates $\sum_{i\in [n]} x_i v_i v_i^{\top}$, where the weights $x_i$ form the solution to the convex relaxation. Indeed~\cite{Allen-Zhu2017} essentially build on this connection to obtain their results for the $E$-design problem discussed earlier. The Kadison-Singer problem~\cite{MarcusSpielmanSrivastavaKS} is closely related to the minimum eigenvalue problem under a partition matroid constraint. We utilize this connection in Corollary~\ref{cor:ks_alg} to give an algorithmic version of the Kadison-Singer problem for constant dimensions. More generally, the Kadison-Singer problem can be reformulated as showing that the integrality gap of the natural relaxation of the minimum eigenvalue problem under partition matroid constraints is at most $1/(1 -\epsilon)$ if the length of each vector is at most $O(\epsilon)$. We discuss this connection in Section~\ref{sec:conclusion}. 

%%%%%%%%%%%%%%%%%%%%%%%%%%%%%%%%%%%%%%%%%%%%%%%%%%%%%%%%%%%%
\section{The Algorithm for Partition Matroids}\label{sec:partition}
To highlight the main idea of our algorithm, we first prove Theorem~\ref{thm:min_lambda} for the special case of partition matroid. 
Let $\mathcal{M} = (E, \mathcal{I})$ be a partition matroid where $E = P_1 \cup \cdots \cup P_k$ be a disjoint union of parts with each part containing $n$ elements, and we have a collection of vectors $v_{ij}$ for $i \in [k]$ and $j \in P_i$. 
The goal is to select an element $\sigma(i) \in P_i$ for each $i$ to maximize
$\lambda_{\min}\left( \sum_{i=1}^k v_{i\sigma(i)}v_{i\sigma(i)}^\top \right).$

We can construct the natural convex relaxation of this problem as follows. For each $i \in [k]$ and $j \in P_i$, we add a decision variable $x_{ij}$ which represents whether we select the vector $v_j$ from part $P_i$, i.e., if $\sigma(i) = j$. Then we get the convex program
\begin{equation*}
    \begin{array}{cl}
        \max & \lambda_{\min}\left( X\right) \\
        & X = \sum\limits_{i=1}^k \sum\limits_{j\in P_i} x_{ij} \cdot v_{ij}v_{ij}^\top \\
        & \sum\limits_{j \in P_i} x_{ij} = 1, \quad \forall i \in [k]\\
        & x \geq 0
    \end{array}
\end{equation*}

The constraint $\sum_{j \in P_i} x_{ij} = 1$ ensures that we have a probability distribution over the possible assignments within each part in the optimal solution.

Given an optimal solution $x^\star$ with value $\OPT$, a natural rounding strategy is to round independently within each part. Following this rounding strategy, we get a rank $1$ random matrix $\bM_i$ for each part $P_i$ with
\begin{equation*}
    \Prob(\bM_i = v_{ij}v_{ij}^\top) = x^\star_{ij}\,, \quad \forall  j \in P_i \,.
\end{equation*}
The following matrix concentration inequality bounds the probability of failure of this rounding strategy.

\begin{theorem}~\cite[Theorem 5.1.1]{tropp2015introduction}
    Consider independent random matrices $\bM_1,\ldots, \bM_k \in \mathbb{S}_d^+$. Set 
    \[\mu_{\min} = \lambda_{\min}\left(\E\left[ \sum_{i=1}^k \bM_i \right]\right)\,.\]
    If $\lambda_{\max}(\bM_i) \leq R$ for all $i \in [k]$ a.s. then
    \[ \Prob\left(\lambda_{\min}\left( \sum_{i=1}^k \bM_i\right) < (1-\epsilon) \mu_{\min}\right) \leq d\cdot \exp\left(\frac{-\epsilon^2 \mu_{min}}{2R}\right).\]
\end{theorem}

If we round according to the optimal solution $x^\star$ then $\E\left[ \sum_{i=1}^k \bM_i \right] = \sum_{i=1}^k \sum_{j \in P_i} x_{ij}^\star v_{ij}v_{ij}^\top$.

So $\mu_{\min} = \OPT$, and since for our particular case $M_i$ are rank $1$, $R = \max_i \lambda_{\max}(M_i) = \max_{ij} \|v_{ij}\|^2$. To bound the failure probability, we want $R \approx \epsilon^2/\log(d)$, which in turn requires that $ \max_{ij} \|v_{ij}\|^2 = O(\epsilon^2/\log(d))$. This is a very strong assumption on an instance.

The plan is to ``guess" a suitable change of basis such that all the vectors in the support of our optimal solution have a small norm. This will be useful because of the following standard, but slightly more flexible, version of the preceding matrix concentration inequality.

\begin{corollary}\label{cor:flexible-matrix-chernoff}
    Consider independent random matrices $\bM_1,\ldots, \bM_k \in \mathbb{S}_d^+$ and let $A$ be an arbitrary positive definite matrix. Define 
    $\mu_{\min} := \lambda_{\min}\left(A^{-1/2}\E\left[ \sum_{i=1}^k \bM_i \right]A^{-1/2}\right).$
    If $\lambda_{\max}(A^{-1/2}\bM_iA^{-1/2}) \leq R$ for all $i \in [k]$ a.s. then
    \[ \Prob\left( \sum_{i=1}^k \bM_i \nsucceq (1-\epsilon) \mu_{\min}\cdot A \right) \leq d\cdot \exp\left(\frac{-\epsilon^2 \mu_{min}}{2R}\right).\]
\end{corollary}

Again, since $\bM_i$ is rank $1$ for our case, we have
$
    R = \max_{i\in [k]} \; \lambda_{\max}(A^{-1/2}\bM_iA^{-1/2}) = \max_{i,j} \; v_{ij}^\top A^{-1} v_{ij} \,.
$
So, to use this corollary, we first need to find a matrix $A$ such that $v_{ij}^\top A^{-1} v_{ij} = O(\epsilon^2/\log(d))$ for all $[i] \in [k], j \in P_i$. 
 We will only need to consider matrices $A$ of a specific form that uses the input vectors.

Given a subset $S\subseteq E$, we define $A_S := \sum_{(i,j) \in S} v_{ij}v_{ij}^\top$, and consider the set of long vectors in the norm induced by $A_S$:
$L(S) := \left\{ (i,j)\in E\del S : v_{ij}^\top A_S^{-1} v_{ij} > \frac{\epsilon^2}{10\log(d)} \right\}.$
For a fixed set $S$, the following convex program ensures that $S$ is included in the solution and no ``long" vectors from $L(S)$ are included in the solution.
\begin{equation*}
    \begin{array}{cl}
        \max & \lambda_{\min}\left( X\right) \\
        & X = \sum\limits_{i=1}^k \sum\limits_{j\in P_i} x_{ij} \cdot v_{ij}v_{ij}^\top \\
        & \sum\limits_{j \in P_i} x_{ij} = 1, \,\, \forall i \in [k]\\
        & x_{ij} = 0, \,\, \forall (i,j) \in L(S)\\
        & x_{ij} = 1, \,\, \forall (i,j) \in S\\
        & x \geq 0
    \end{array}\tag{CP(S)}\label{eq:partition_convex_program}
\end{equation*}
Because of the extra constraints excluding ``long'' vectors, we could now use the flexible matrix concentration inequalities to randomly round the optimal solution. 

But, it is not clear that there is a good choice of $S$ for which the convex program~\ref{eq:partition_convex_program} is still a relaxation of the original problem. Lemma~\ref{lem:local_search}, which we restate here for the reader's convenience, shows that there exists a suitable set $S$ that is not too large.

\localsearch*

The proof of this lemma is inspired by the local search algorithm of~\cite{madan2019}.

At first glance, it may not be apparent why a subset satisfying the conditions of Lemma~\ref{lem:local_search} should exist. However, in the proof, we show that any subset of $T$ that is locally optimal with respect to a local search criteria indeed satisfies the guarantees of Lemma~\ref{lem:local_search}.
\begin{proof}{(of Lemma~\ref{lem:local_search})}
    We consider the local search process of~\cite{madan2019}. Starting with a set $S$ of size $\ell$ such that $A = \sum_{i\in S} v_i v_i^\top$ is invertible, we apply the following update rule. For any $j \in T \backslash S$ and $i \in S$, if $ \det(A) < \det(A - v_iv_i^\top + v_jv_j^\top)$, update $S = \{S\backslash \{i\}\} \cup \{j\}$ and iterate.
    
    Let $S \subseteq T$ be a locally optimal (under single element swaps) solution for this process (such an $S$ corresponds to the locally optimal solution determinant maximization problem subject to the cardinality constraint $|S| \leq \ell$), and let $A = \sum_{i\in S} v_i v_i^\top$. More concretely, this means that for all $i \in S$ and $j \in T\del S$,
    \[ \det(A) \geq \det(A - v_iv_i^\top + v_jv_j^\top) \,.\]
    
    We calculate the determinant on the right-hand side using the matrix determinant lemma,
    \begin{align*}
        \det(A - v_iv_i^\top + v_jv_j^\top)
        &= \det\left( A + \bmat{v_i & v_j} \bmat{-v_i & v_j}^\top\right) = \det(A) \cdot \det\left( I_2 + \bmat{-v_i & v_j}^\top A^{-1} \bmat{v_i & v_j} \right) \\
        &= \det(A) \cdot \left( (1- v_i^\top A^{-1} v_i)(1 + v_j^\top A^{-1}v_j) + (v_i^\top A^{-1} v_j)^2 \right)\,.
    \end{align*}
     So local optimality implies that for every $i \in S$ and $j \notin S$, $(1-v_i^\top A^{-1} v_i)(1+v_j^\top A^{-1} v_j) + (v_i^\top A^{-1} v_j)^2 \leq 1$.
     Rearranging this inequality we get
    \begin{equation}
        v_j^\top A^{-1} v_j- (v_i^\top A^{-1} v_i) \cdot (v_j^\top A^{-1} v_j) +  (v_i^\top A^{-1} v_j)^2 \leq v_i^\top A^{-1} v_i \,. \label{eq:local_opt}
    \end{equation} 
    Note that $\sum_{i\in S} v_i^\top A^{-1} v_i = \langle A, A^{-1} \rangle = d$ and $\sum_{i\in S} (v_i^\top A^{-1} v_j)^2 = v_j^\top A^{-1} v_j$. So for a fixed $j \in T\backslash S$, summing equation~\eqref{eq:local_opt} over all $i \in S$ implies $\ell \cdot v_j^\top A^{-1} v_j - d \cdot v_j^\top A^{-1} v_j + v_j^\top A^{-1} v_j \leq d$.
    Rearranging, we see that for any $j \in T \backslash S$,
    \[ v_j^\top A^{-1} v_j \leq \frac{d}{\ell-d+1} = \frac{\epsilon^2}{10\log(d)}\,,\]
where the last equality follows by choosing $\ell = 10d\log(d)/ \epsilon^2 + d -1$.
\end{proof}

We will apply this lemma to the case when $T = \{ v_{i\sigma^\star(i)} : i \in [k]\}$m where $\sigma^*$ is the choice function that maximizes the minimum eigenvalue, i.e., when $T$ contains the vectors from an optimal integral assignment. In particular, we get the following corollary.

\begin{lemma}\label{lem:S_relaxation}
    There is a subset $S\subseteq E$ such that $|S| = O( d\log(d) /\epsilon^2 )$, and the convex program~\ref{eq:partition_convex_program} is a relaxation for the minimum eigenvalue problem.
\end{lemma}

As $d$ is a constant, the size of the set $S$ we search for is also constant. Thus, there are at most $O(n^{O(d\log(d)/\epsilon^2)})$ possible choices for $S$. We will consider each choice in turn to guess the correct set. Note that trying every set of the appropriate size will be the dominant factor in determining the algorithm's runtime.

The following lemma proves that for any fixed subset $S$, rounding the optimal solution to~\ref{eq:partition_convex_program} gives a good approximation to the optimal value of~\ref{eq:partition_convex_program}. 

\begin{lemma}\label{lem:partition_CP_rounding}
    Let $S \subseteq E$ be an independent set, and let $x^\star$ be the optimal solution to~\ref{eq:partition_convex_program}. Then rounding randomly in each part outputs an assignment $\sigma: [k] \rightarrow E$ with $\sigma(i) \in P_i$ such that
    \[ \Prob\left[\lambda_{\min}\left(\sum_{i=1}^k v_{i\sigma(i)}v_{i\sigma(i)}^\top\right) < (1-\epsilon) \cdot \lambda_{\min}\left( \sum_{i=1}^k \sum_{j\in P_i} x^\star_{ij}\, v_{ij}v_{ij}^\top \right) \right] < d^{-4}.\]
\end{lemma}

\begin{proof}
    Let $X =  \sum_{i=1}^k \sum_{j\in P_i} x^\star_{ij}\, v_{ij}v_{ij}^\top $. The matrix $X^\star$ contains $v_{ij}v_{ij}^\top$ with coefficient $1$ for every $(i,j) \in S$. Thus $X \succeq \sum_{i\in S} v_iv_i^\top$, so $v_j^\top X^{-1} v_j \leq \frac{\epsilon^2}{10\log(d)}$ for all $j \notin L\cup S$. 
    For the purposes of the analysis, for every $j \in S$ we can replace the vector $v_i$ with $(v_i^\top X^{-1} v_i) \cdot \sqrt{10\log(d)/\epsilon^2}$ copies of the same vector scaled down to have squared-length at most $\epsilon^2 /(10\log(d))$ with respect to $X$. Since all elements of $S$ get value $1$ in $x^\star$, we can similarly extend the vector $x^\star$ so that it has a $1$ in all the copied entries. Since these values of $x^*$ are deterministic, nothing changes about the resulting distribution over matrices, but we can now assume that $v_{ij}^\top X^{-1} v_{ij} < \epsilon^2 /(10\log(d))$ for all $(i,j)$ in the support of $x^\star$.

    Next, define random matrices $\bM_1,\ldots,\bM_k$ such that for any $i \in [k]$, $\Prob\left(M_i = v_{ij} v_{ij}^\top \right) = x^\star_{ij}$ for all $j \in P_i$. We then apply Corollary~\ref{cor:flexible-matrix-chernoff} with $A = X$ on random matrices $\bM_1,\ldots,\bM_k$. Since $x^\star$ is not supported on $L$, 
    \begin{equation*}
        R = \max_i \; \lambda_{\max}(X^{-1/2}\bM_i X^{-1/2}) \leq \frac{\epsilon^2}{10\log(d)}\,.
    \end{equation*}
    In addition, as $\E\left(\sum_i M_i \right)= X$ by definition, we have $\mu_{\min} = 1$.

    Thus, if $\sigma: [k] \rightarrow E$ is the choice function obtained by independent rounding, 
    \begin{align*} 
        \Prob_{\sigma} \left[ \sum_{i=1}^k v_{i\sigma(i)}v_{i\sigma(i)}^\top \nsucceq  (1-\epsilon) X\right] 
        &\leq d \cdot \exp(-5\log(d)) = d^{-4}.
    \end{align*}
    We conclude that $\lambda_{\min}\left(\sum_{i=1}^k v_{i\sigma(i)}v_{i\sigma(i)}^\top \right) \geq (1-\epsilon)\lambda_{\min}\left( X \right) $ with probability at least $1 - d^{-4}$.
\end{proof}
 
Combining this lemma with the earlier guarantee that there exists a set $S$ of reasonable size such that~\ref{eq:partition_convex_program} is a relaxation, we get the following algorithm: try all possible choices for the set $S$ and return the solution with the best objective.
\begin{algorithm}[H]
    \caption{Algorithm to find an approximation to $\OPT$}
    \label{alg:partition_alg}
    \begin{algorithmic}[1]
        \State \textbf{Input}: Partition matroid $\mathcal{M}$ with $k$ parts $P_1, \ldots, P_k$
        \For{each $S \subseteq [n]$ such that  $|S| = 10d\log(d)/ \epsilon^2 + d -1$}
        \State $x^* \leftarrow$ optimal solution of~\ref{eq:partition_convex_program} for matroid $\mathcal{M}$
        \State For each $i \in [k]$, assign $\sigma_S(i) = j$ with probability $x^*_{ij}$ 
        \EndFor
        \State Return the choice function $\sigma_S$ which maximizes $  \lambda_{\min} \left(\sum_{i} v_{i\sigma_S(i)} v_{i\sigma_S(i)}^\top\right)$ over all choices of $S$
    \end{algorithmic}
\end{algorithm}

\begin{proof}{(of Theorem~\ref{thm:min_lambda} for Partition Matroids)}
    By Lemma~\ref{lem:S_relaxation} there is a set $S\subseteq E$ with $|S| = O(d\log d/ \epsilon^2)$ such that~\ref{eq:partition_convex_program} is a relaxation.

    Let $x^*$ be the optimal value of~\ref{eq:partition_convex_program} and let $ \sigma^\star$ be the choice function of the optimal basis for the  minimum eigenvalue problem.
    Since~\ref{eq:partition_convex_program} is a relaxation, we have $ \lambda_{\min}\left(\sum_{i =1}^k v_{i\sigma^*(i)} v_{i\sigma^*(i)}^\top \right) \leq  \lambda_{\min}\left(\sum_{ij} x^*_{ij} v_{ij} v_{ij}^\top \right)$.
    
    Lemma~\ref{lem:partition_CP_rounding} implies that with high probability, the choice function obtained by rounding $x^*$, $\sigma_S$, is a good approximation to~\ref{eq:partition_convex_program}. So combining Lemma~\ref{lem:partition_CP_rounding} with the previous inequality gives
    \begin{align*}
        \lambda_{\min}\left( \sum_{i=1}^k v_{i\sigma_S(i)}v_{i\sigma_S(i)}^\top \right)
        &\geq (1-\epsilon)\cdot \lambda_{\min}\left(\sum_{ij} x^*_{ij} v_{ij} v_{ij}^\top \right)
        \geq (1-\epsilon)\cdot \lambda_{\min}\left(\sum_{i =1}^k v_{i\sigma^*(i)} v_{i\sigma^*(i)}^\top \right)\,,
    \end{align*}
    with probability at least $1 - d^{-4}$. Since we iterate over all choice functions in step 6 of Algorithm~\ref{alg:partition_alg}, we will output a choice function $\sigma$ which is at least as good as $\sigma_S$ with the same probability.
\end{proof}
\subsection{Application: Algorithmic Kadison-Singer Problem} \label{sec:ks}
The Kadison-Singer conjecture was resolved in~\cite{MarcusSpielmanSrivastavaKS} using the following theorem which can be interpreted as a generalization of Weaver's conjecture ~\cite{weaver2004kadison}.

\begin{theorem}\cite[Corollary 1.5 with $r=2$]{MarcusSpielmanSrivastavaKS}
    Let $u_1,\ldots,u_m\in\R^d$ be vectors such that $\sum_{i=1}^m u_iu_i^\top = I$ and $\|u_i\|^2 \leq \alpha$ for all $i.$ There exists a set $T \subseteq [m]$ such that
    \begin{equation*}
        \left(\frac{1}{2} - 3\sqrt{\alpha}\right)I_d \preceq \sum_{i \in T} u_iu_i^\top \preceq \left(\frac{1}{2} + 3\sqrt{\alpha}\right) I_d.
    \end{equation*}
\end{theorem}

Their proof is based on analyzing interlacing families of polynomials and does not lead to an efficient algorithm to find such a subset $T$.

In~\cite{jourdan2022algorithmic}, they introduce an algorithmic form of the Kadison-Singer problem, which asks to find such a subset assuming it exists. For a constant $c > 0$ and a set of vectors $u_1,\ldots, u_m \in \R^d$ such that $\|u_i\|^2 \leq \alpha$, $\sum_{i=1}^m u_iu_i^\top = I$ where there exists a subset $T\subseteq [m]$ satisfying 
\begin{equation}
    \left(\frac{1}{2} - c\sqrt{\alpha}\right)I \preceq \sum_{i \in T} u_iu_i^\top \preceq \left(\frac{1}{2} + c\sqrt{\alpha}\right),\label{ks:feasibility}
\end{equation}
the goal is actually to find a set $T\subseteq [m]$ which satisfies the above condition. This problem is FNP-hard when $c = 1/(4\sqrt{2})$ for general values of $d$~\cite[Theorem 2]{jourdan2022algorithmic}.

Their main result~\cite[Theorem 1]{jourdan2022algorithmic} is an algorithm with running time 
\[O\left( \binom{m}{k} \cdot \text{poly}(m,d)\right)\text{ for }k = O\left( \frac{d}{\epsilon^2} \log(d) \log\left( \frac{1}{c\sqrt{\alpha}}\right)\right)\,,\]
which returns a set $T' \subseteq [m]$ such that  
\begin{equation}
    (1-\epsilon)\left(\frac{1}{2} - c\sqrt{\alpha}\right)I \preceq \sum_{i \in T'} u_iu_i^\top \preceq (1+\epsilon)\left(\frac{1}{2} + c\sqrt{\alpha}\right)I, \label{ks:approximate_feas}
\end{equation}

In this section, we will show how to use the rounding technique for partition matroids to give a simpler algorithm that achieves the same guarantee with the same run time, except we save the small dependence on $\log(1/c\sqrt{\alpha})$ in the exponent.

\begin{proof}{(of Corollary~\ref{cor:ks_alg})}
    Given vectors $u_1,\ldots, u_m \in \R^d$, we construct an instance of the minimum eigenvalue with partition constraints as follows. Let $E = \{1,2\} \times [m]$, with $m$ parts $P_1,\ldots, P_m$ so that $P_i = \{(i,1),(i,2)\}$ for $i \in [m]$. For each $i \in [m]$ define the vectors
    \[ v_{i1} = \bmat{u_i\\0} \in \R^{2d}, \text{ and } v_{i2} = \bmat{0\\u_i} \in \R^{2d}.\]
    To see how $v$ and $u$ are related, note that for any $\delta \in [0,1/2)$ there is a choice function $\sigma: [m] \rightarrow \{1,2\}$ such that
    \begin{equation}
        \left(\frac{1}{2} - \delta\right) I_{2d} \preceq \sum_{i=1}^m v_{i\sigma(i)}v_{i\sigma(i)}^\top \label{eq:e1}
    \end{equation}
    if and only if there is a set $T\subseteq [m]$ such that
    \begin{equation}
        \left(\frac{1}{2} - \delta\right)I \preceq \sum_{i \in T} u_iu_i^\top \preceq \left(\frac{1}{2} + \delta\right)I. \label{eq:e2} 
    \end{equation} 
    Given $\sigma$ satisfying~\eqref{eq:e1}, let $X_1 := \sum_{i:\sigma(i) =1} u_iu_i^\top$ and $X_2 := \sum_{i : \sigma(i) = 2} u_iu_i^\top$. Then $X_1$ and $X_2$ are respectively the first and second diagonal $d\times d$ block of $\sum_{i=1}^m v_{i\sigma(i)}v_{i\sigma(i)}^\top$. Therefore $\left( \frac{1}{2} - \delta\right) I_{2d} \preceq \sum_{i=1}^m v_{i\sigma(i)}v_{i\sigma(i)}^\top$ if and only if $X_1 \succeq \left( \frac{1}{2} - \delta\right) I_{d}$ and $X_2 \succeq \left( \frac{1}{2} - \delta\right) I_{d}$. In addition, since $X_1 + X_2 = I_d$, this is equivalent to
    \begin{equation*}
        \left(\frac{1}{2} - \delta\right)I_{d} \preceq \sum_{i: \sigma(i)=1} u_i u_i = I_d - X_2 \preceq \left(\frac{1}{2} + \delta\right)I_{d}.
    \end{equation*}

    We then use Algorithm~\ref{alg:partition_alg} to find a $(1-\epsilon)$ approximate solution $\sigma: [m] \rightarrow \{1,2\}$ to input $\mathcal{M}$ and vectors $v_{ij}$. 
    Since we assume there is a set $T$ satisfying~\eqref{ks:feasibility}, Theorem~\ref{thm:min_lambda} implies that with probability at least $1-O(d^{-4})$, Algorithm \ref{alg:partition_alg} will return a choice function $\sigma^*$ such that $(1-\epsilon)\left(\frac{1}{2} - c\sqrt{\alpha}\right)I_{2d} \preceq \sum_{i=1}^m v_{i\sigma^*(i)}v_{i\sigma^*(i)}^\top$, and we will return the set  $T' = \{i \in [m] : \sigma^*(i) = 1\}$.

    From the equivalence between~\eqref{eq:e1} and~\eqref{eq:e2}, the set $ T' = \{i \in [m] : \sigma(i) = 1\} $ satisfies~\eqref{ks:approximate_feas} 

    \begin{equation*}
        (1-\epsilon)\left(\frac{1}{2} - c\sqrt{\alpha}\right)I_{d} \preceq \sum_{i \in T'} u_i u_i  \preceq (1+\epsilon)\left(\frac{1}{2} + c\sqrt{\alpha}\right)I_{d}\,.
    \end{equation*}
\end{proof}

%%%%%%%%%%%%%%%%%%%%%%%%%%%%%%%%%%%%%%%%%%%%%%%%%%%%%%%%%%%%
\section{General Matroid Constraints}

In the general form of the problem, we are given a collection of vectors $v_1,\ldots, v_n\in\R^d$ and a matroid $\M = ([n],\I)$, and the goal is to find a basis $B \in \I$ which maximizes $ \lambda_{\min}\left( \sum_{i \in B} v_iv_i^\top \right).$
For background on matroids, see Appendix~\ref{sec:matroids}.

For a general matroid, the idea of finding a linear transformation under which all elements in the optimal solution have a small norm generalizes easily. So we can use the same approach of first guessing a set $S\subseteq E$ on a reasonable size and then solving the convex relaxation of the problem conditioned on $S$ being included in the solution.

Given a subset $S\subseteq [n]$, we can again set $A_S = \sum_{i \in S} v_{i}v_{i}^\top$, and consider the set of long vectors:
\[ L(S) = \left\{i \in [n]\del S : v_{i}^\top A_S^{-1} v_{i} > \frac{\epsilon^2}{10\log(d)}\right\}.\]
For a matroid $\M$, let $\mathcal{P}(M)\subseteq [0,1]^n$ denote the matroid base polytope. Then the following is the natural convex programming relaxation which excludes the ``long'' vectors.
\begin{equation*}
    \begin{array}{cl}
        \max & \lambda_{\min}\left( X\right) \\
        & X = \sum\limits_{i=1}^n x_{i} \cdot v_{ij}v_{ij}^\top \\
        & x \in \mathcal{P}(\M) \\
        & x_{i} = 0, \,\, \forall i \in L(S)\\
        & x_{i} = 1, \,\, \forall i \in S\\
        & x \geq 0
    \end{array}\tag{CP(S)}\label{eq:matroid_convex_program}
\end{equation*}
This convex program can be solved in polynomial time (see Appendix~\ref{sec:matroids}).
Just like in the partition case, Lemma~\ref{lem:local_search} guarantees that there is a set $S$ for which~\ref{eq:matroid_convex_program} is a relaxation for the minimum eigenvalue problem. 
As before, after solving~\ref{eq:matroid_convex_program}, we can guarantee that all the vectors in the fractional support of the optimal solution will have a small norm with respect to $A_S$. 

The challenge in extending the earlier approach to general matroid constraints comes from the rounding step. For a partition matroid, we could simply round the fractional optimum of~\ref{eq:partition_convex_program} independently in each part to obtain a basis. However, for more general constraints, it is not so clear how to round a fractional solution to a basis. 

Instead of rounding independently, we will use the technique of pipage rounding to find a basis.  The following lemma is the lower-tail version of the same concentration inequality proved in~\cite{harvey2014pipage}. For completeness, we will include a proof of the version we need in Appendix~\ref{sec:pipage_rounding_appendix}.

\begin{lemma}\label{lem:renormalized-pipage}
    Let $\mathcal{P}(\mathcal{M})$ be a matroid base polytope and $x \in \mathcal{P}(\mathcal{M})$. Let $M_1,\ldots, M_m$ be self-adjoint matrices that satisfy $\lambda_{\max}(M_i)\leq R $. Let $\mu = \lambda_{\min}\left(\sum_{i \in [n]} x_i M_i \right)$.  If randomized pipage rounding (Algorithm~\ref{alg:pipage_rounding}) starts at $x$ and
outputs the extreme point $\hat{x} = \chi(B)$ of $\mathcal{P}(\mathcal{M})$, then we have
\begin{equation*}
    \Prob\left[ \sum_{i \in B} M_i \leq (1-\epsilon) \cdot \mu \right] \leq  d\cdot \exp\left(\frac{-\epsilon^2 \mu}{2R}\right) \,.
\end{equation*}

\end{lemma}

We use this lemma to generalize our earlier approach to all matroids.

\begin{lemma} \label{lem:subset_opt}
    Let $S \subseteq E$ be an independent set in $\mathcal{M}$ and let $x^\star$ be the optimal solution to~\ref{eq:matroid_convex_program}. Then pipage rounding starting at $x^\star$ outputs a basis $B$ such that
    \[ \Prob\left[\lambda_{\min}\left(\sum_{i\in B} v_iv_i^\top\right) < (1-\epsilon) \lambda_{\min}\left( \sum_{i \in [n]} x^\star_i\, v_iv_i^\top \right) \right] < d^{-4}.\]
\end{lemma}

The proof is identical to that of Lemma~\ref{lem:partition_CP_rounding}, except we use the matrix concentration inequality from Lemma~\ref{lem:renormalized-pipage}.

Using this lemma, the following algorithm gives a $(1-\epsilon)$-approximation with high probability.

\begin{algorithm}[H]
    \caption{Algorithm to find an approximation to $\OPT$}
    \label{alg:general_algorithm}
    \begin{algorithmic}[1]
        \For{each $S \subseteq [n]$ such that  $|S| = 10d\log(d)/ \epsilon^2 + d -1$}
        \State Solve~\ref{eq:matroid_convex_program} to get optimal solution $x^\star$.
        \State Let $B_S \leftarrow$ basis returned by Algorithm~\ref{alg:pipage_rounding} for input  $x^\star$ 
        \EndFor
        \State Return the basis $B_S$ with the best objective
    \end{algorithmic}
\end{algorithm}
\subsection{Proof of Theorem 1 and Theorem 2}
In this section we prove Theorem~\ref{thm:general_f}. The proof of Theorem~\ref{thm:min_lambda} follows identically.
\begin{proof}{(of Theorem~\ref{thm:general_f})}
Let $B^* = \arg\max_{B \in \I} f(\sum_{i \in B} v_i v_i^\top)$ and let $OPT = f\left(\sum_{i \in B^*}v_i v_i^\top\right)$. Let  $S  \subseteq B^*$ such that $|S| = O(d\log{d}/\epsilon^2)$, $A = \sum_{(i,j) \in S} v_{ij} v_{ij}^\top$ is invertible, and for all $(i,j)\in B^* \backslash S$, $v_{ij}^\top A^{-1} v_{ij} \leq \epsilon^2/10d\log{d}$. By Lemma~\ref{lem:local_search}, such a set $S$ exists.

Let $x^*$ be the optimal solution of~\ref{eq:matroid_convex_program}. Since the indicator vector of $B^*$ satisfies the constraints of~\ref{eq:matroid_convex_program}, $ OPT = f\left(\sum_{i \in B^*} v_i v_i^\top \right) \leq f\left(\sum_{i \in E} x^*_i \cdot v_i v_i^\top \right)$. Therefore,
\begin{equation*}
    \Pr\left[ f(\sum_{i \in \Tilde{B}} v_i v_i^\top ) < (1-\epsilon) \cdot OPT \right] \leq  \Pr\left[ f(\sum_{i \in \Tilde{B}} v_i v_i^\top ) < (1-\epsilon) \cdot f\left(\sum_{i \in E} x^*_i \cdot v_i v_i^\top \right)\right]\,.
\end{equation*}
Let $X :=\sum_{i \in E} x^*_i \cdot v_i v_i^\top $. Since $f$ is monotone and homogeneous, we have
    \begin{align}
        \Pr\left[ f(\sum_{i \in \Tilde{B}} v_i v_i^\top ) < (1-\epsilon) \cdot f(X)\right] &= \Pr\left[ f(\sum_{i \in \Tilde{B}} v_i v_i^\top ) < f((1-\epsilon) \cdot X)\right]\\
        &\leq \Pr\left[ \sum_{i \in \Tilde{B}} v_i v_i^\top  \nsucceq (1-\epsilon) \cdot X \right] \notag\\
        &= \Pr\left[ \sum_{i \in \Tilde{B}} X^{-1/2}v_i v_i^\top X^{-1/2}  \nsucceq (1-\epsilon) \cdot I \right] \notag \\
        &= \Pr\left[ \lambda_{\min}(\sum_{i \in \Tilde{B}} X^{-1/2}v_i v_i^\top X^{-1/2})  \leq (1-\epsilon)  \right] \label{eq:ineq1}
    \end{align} 
    
     Similar to the proof of Lemma~\ref{lem:partition_CP_rounding}, we will apply Lemma~\ref{lem:renormalized-pipage} to random matrices $M_i = v_i v_i^\top$ after appropriate transformations to ensure $R = O(\epsilon^2/\log{d})$. First note that since $x^*_i = 1$ for any $i \in S$, we have $i \in B$ as pipage rounding does not change the integral elements of $x^*$. Therefore $\sum_{i \in S} v_i v_i^\top \preceq X$, and for any $i \in B\backslash S$,
   \begin{equation*}
       \lambda_{\max}(X^{-1/2} v_i v_i^\top X^{-1/2}) = v_i^\top X^{-1} v_i \leq v_i^\top \left( \sum_{j\in S} v_j v_j^\top \right)^{-1} v_i \leq \frac{\epsilon^2}{10\log{d}}.
   \end{equation*}
   For any $i \in S$, $v_i X^{-1} v_i^\top \leq 1$, and since $x_i^* = 1$ we can replace $M_i = v_i v_i^\top$ with $r = 10\log{d}/ \epsilon^2$ matrices $M_{i}^1, \ldots, M_{i}^r$ such that  $M_i^j = \frac{\epsilon^2}{10\log{d}} v_i v_i^\top$. This ensures that $\lambda_{\max}(X^{-1/2} M_{i}^j X^{-1/2}) \leq \epsilon^2/10\log{d}$ for all $j \in [r]$. Applying Lemma~\ref{lem:renormalized-pipage} on random matrices $\{M_i\}_{i\in B\backslash S}, \{M_i^j\}_{i \in S, j \in [r]}$ gives 

   \begin{equation*}
        \Pr\left[ \lambda_{\min}(\sum_{i \in \Tilde{B}} X^{-1/2}v_i v_i^\top X^{-1/2})  \leq (1-\epsilon)  \right] \leq d^{-4}\,.
   \end{equation*}

   If $B$ is the basis returned by Algorithm~\ref{alg:general_algorithm}, then $f(B) \geq f(\Tilde{B})$. Using the above inequality with~\eqref{eq:ineq1} implies that $f(\sum_{i \in B} v_i v_i^\top ) \geq (1-\epsilon) \cdot OPT$ with probability at least $1-d^{-4}$.
\end{proof}

\subsection{Pipage Rounding}\label{sec:pipage_rounding}

The purpose of this section is to give an explanation of the pipage rounding technique and motivate the proof of Lemma~\ref{lem:renormalized-pipage}.
For a detailed discussion of pipage rounding, see~\cite{harvey2014pipage}. 

For a set $S\subseteq E$, let $x^*$ be the optimal solution of~\ref{eq:matroid_convex_program}. If we round $x^*$ independently, i.e., add element $i$ to the output with probability $x_i^*$, then the set obtained, say $B$, might not be independent in $\mathcal{M}$. But the following concentration inequality would still hold due to independence,
\begin{equation}
     \Prob\left[\lambda_{\min}\left(\sum_{i\in B} v_iv_i^\top\right) < (1-\epsilon) \lambda_{\min}\left( \sum_{i \in [n]} x^\star_i\, v_iv_i^\top \right) \right] < d^{-4}. \label{eq:ind}
\end{equation}

The main idea behind pipage rounding is to iteratively transform a point $x \in \mathcal{P}(\mathcal{M})$ to a basis of $\mathcal{M}$ while ensuring that the failure probability from equation~\eqref{eq:ind} does not increase.

For a point $x \in [0,1]^n$, let $D(x)$ represent the corresponding product distribution over $\{0,1\}^n$ with marginals given by $x$, i.e., include element $i$ in the output with probability $x_i$.
For $x \in [0,1]^n$ and $\epsilon > 0$, define 
\begin{equation*}
    p_\epsilon(x) := \Pr_{B \sim D(x)}\left( \lambda_{\min}(\sum_{i \in B}v_i v_i^\top) \leq (1-\epsilon)\cdot \lambda_{\min}\left(\sum_{i\in [n]} x_i v_i v_i^\top \right) \right)\,.
\end{equation*}
So $p_\epsilon(x)$ is the failure probability of getting a $(1-\epsilon)$-approximation when rounding independently at point $x$. \cite{harvey2014pipage} showed that there exists a function $g_\epsilon(x)$ s.t. $p_\epsilon(x) \leq g_\epsilon(x) \leq d\cdot \exp\left(\frac{-\epsilon^2 \mu_{min}}{2R}\right) $ and $g_\epsilon$ is concave under swaps, i.e., for all $a,b \in [n]$ and $x \in \mathcal{P}(\mathcal{M})$ the map 
   $ z \mapsto g_\epsilon(x + z(e_a - e_b))$
is concave.

So, if $x$ is not an extreme point of $\mathcal{P}(\mathcal{M})$, then there exist $a,b \in [n]$ and $\epsilon > 0$ such that $x \pm \epsilon(e_a + e_b) \in \mathcal{P}(\mathcal{M})$. Let $l = \min\{z: x + z(e_a-e_b) \in \mathcal{P}(\mathcal{M})\}$ and $ u = \max\{z: x + z(e_a-e_b) \in \mathcal{P}(\mathcal{M})$.

With this, we can define $x^l = x + l(e_a - e_b)$ and $x^u = x + u(e_a - e_b)$. Since $g(x + z(e_a-e_b))$ is concave as a function of $z$, we know that either $g(x^l) \leq g(x)$ or $g(x^u)\leq g(x)$. Moreover, both $x^l$ and $x^u$ are on a lower dimensional face than the initial point $x$. Thus, for any initial point $x_0 \in \mathcal{P}(\mathcal{M})$, a total of $m$ iterations suffice to find an extreme point with $\hat x$ with $g(\hat x) \leq g(x_0)$.

In randomized pipage starting at $x \in \mathcal{P}(\mathcal{M})$, our next iterate $x'$ of the rounding procedure will be $x^l$ with probability $\frac{u}{u-l}$ and $x^u$ with probability $\frac{-l}{u-l}$. This ensures that $\E(x') = x$, and the concavity under swaps guarantees that $\E[g_\epsilon(x')] \leq g_\epsilon(x)$ by Jensen's in the variable $z$. If we start at a point $x_0 \in \mathcal{P}(\mathcal{M}))$ and iterate this random procedure $m$ times, we get an extreme point $\hat x$ which satisfies $\E[\hat x] = x_0$ and $\E[g(\hat x)] \leq g(x_0)$.

This gives the intuition behind the proof of Lemma~\ref{lem:pipage_random_alg}, and leads to the following algorithm.
\begin{algorithm}[H]
    \caption{Randomized Pipage Rounding}
    \label{alg:pipage_rounding}
    \begin{algorithmic}[1]
    \State \textbf{Input}: Point $x \in \mathcal{P}(\mathcal{M})$, where $\mathcal{P}(\mathcal{M})$ is a matroid base polytope
    \While{$x$ is not integral}
    \State $a,b\leftarrow$ distinct elements of $[n]$ s.t. $\exists \epsilon > 0$ with $x \pm \epsilon (e_a - e_b) \in \mathcal{P}(\mathcal{M})$
    \State $\ell \leftarrow \min\{y \geq 0: x - y(e_a-e_b) \in \mathcal{P}(\mathcal{M})\}$
    \State $h \leftarrow \max\{y \geq 0: x + y(e_a-e_b) \in \mathcal{P}(\mathcal{M})\}$
    \State $x \leftarrow \begin{cases} x - \ell(e_a-e_b) & \text{w.p.  } \ell/(\ell + h)\\
    x + h(e_a-e_b) & \text{w.p.  } h/(\ell + h) 
    \end{cases}$
    \EndWhile
    \State Return basis $B \in \mathcal{P}(\mathcal{M})$ with indicator vector $x$
    \end{algorithmic}
    
\end{algorithm}

\section{Conclusion and Remarks}\label{sec:conclusion}

The resolution of the Kadison-Singer problem in~\cite{MarcusSpielmanSrivastavaKS} using the interlacing families of polynomials implies the following existential result about maximizing the minimum eigenvalue under partition matroid constraints.

\begin{theorem}~\cite[Theorem 1.4]{MarcusSpielmanSrivastavaKS} \label{thm:ks} For $\epsilon > 0$ and vectors $\{v_{ij}\}_{i \in [k], j \in [n]}\in \R^d$ with $\|v_{ij}\|^2 \leq \epsilon$ for all $i \in [k], j \in [n]$, if there exist $x_{ij} \geq 0$ such that
    \begin{equation*}
        \sum_{i=1}^k\sum_{j=1}^n x_{ij} \cdot v_{ij}v_{ij}^\top = I_d \quad \text{ and } \quad \sum_{j=1}^n x_{ij} = 1 \text{ for all }i\in[k],
    \end{equation*}
    then there exists a choice function $\sigma: [k]\rightarrow [n]$ such that
    \begin{equation*}
        (1-\sqrt{\epsilon})^2\cdot I_d \preceq \sum_{i=1}^k v_{i\sigma(i)}v_{i\sigma(i)}^\top \preceq (1+\sqrt{\epsilon})^2 \cdot I_d.
    \end{equation*}
\end{theorem}

We can state this result equivalently as an ``existential" rounding result.
When $\|v_{ij}\|^2 \leq \epsilon$, Theorem~\ref{thm:ks} implies that the integrality gap of the natural convex relaxation~\eqref{eq:cvx_natural} for the minimum eigenvalue problem with partition constraints is only $1/(1-\sqrt{\epsilon})^2$. It is an open problem to efficiently round the solution to the convex relaxation with comparable guarantees for any dimension $d$.

More generally, the problem of designing an approximation algorithm for the minimum eigenvalue problem under partition or matroid constraints in arbitrary dimensions remains wide open. However, checking whether there is a solution with a non-zero objective can be solved in polynomial time solvable through matroid intersection. Recently, there has been significant progress in the case of maximizing the determinant~\cite{J-simplex,DiSumma2015,Det-max-partition,AnariGV18,SinghX18,determinant-maximization,Det-Max2022}, but it remains open whether those techniques can be utilized for the minimum eigenvalue problem.

\newpage
\bibliographystyle{alpha}
\bibliography{paper}

%%%%%%%%%%%%%%%%%%%%%%%%%%%%%%%%%%%%%%%%%%%%%%%%%%%%%%%%%%%%
\appendix
\section{Omitted proofs}\label{app:omitted_proofs}
\begin{proof}{(of Corollary \ref{cor:flexible-matrix-chernoff})}
    This is a simple calculation, using the fact the the semidefinite order is preserved under conjugation.
    \begin{align*}
        \Prob\left( \sum_{i=1}^k \bM_i \nsucceq (1-\epsilon) \mu_{\min}\cdot A \right)
        &= \Prob\left( \sum_{i=1}^k A^{-1/2}\bM_iA^{-1/2} \nsucceq (1-\epsilon) \mu_{\min}\cdot I\right)\\
        &= \Prob\left( \lambda_{\min}\left(\sum_{i=1}^k A^{-1/2}\bM_iA^{-1/2}\right) < (1-\epsilon) \mu_{\min}\right)\\
        &\leq d\cdot \exp\left(\frac{-\epsilon^2 \mu_{min}}{2R}\right)\,.
    \end{align*}
\end{proof}

\subsection{Integrality Gap Example} \label{sec:integrality_gap}
Consider the vectors $v_{1} = e_1, v_{2} = e_1, v_{3} = e_2, v_{4} = e_3$ in $\R^3$ and a partition matroid $\mathcal{M} = ([4], \mathcal{I})$ defined by the bases $\{1,2,3\}, \{1,2,4\}$. The optimal value of maximizing the minimum eigenvalue for this instance is $0$ as we are forced to pick $v_1$ and $v_2$ in any basis and they are linearly dependent.

The convex relaxation of maximizing the minimum eigenvalue for this instance is given by
\begin{equation}
    \begin{array}{cl}
        \max & \lambda_{\min}\left( X\right) \label{eq:cvx_int}\tag{CP}\\
        & X = x_{1} \cdot v_1 v_1^\top + x_{2} \cdot v_2 v_2^\top + x_{3} \cdot v_3 v_3^\top+ x_{4} \cdot v_4 v_4^\top \\
        & x_1 = 1, \quad \forall i \in [k]\\
        & x_2 = 1, \quad \forall i \in [k]\\
        & x_3 + x_4 = 1, \quad \forall i \in [k]\\
        & x \geq 0
    \end{array}
\end{equation}
The optimum of~\eqref{eq:cvx_int} is attained when $x_1 = x_2 = 1$ and $x_3 = x_4$ which gives
\begin{equation*}
    X = 2 e_1e_1^\top + \frac{1}{2} e_2 e_2^\top + \frac{1}{2} e_3 e_3^\top\,.
\end{equation*}
So the optimal value of~\eqref{eq:cvx_int} is $1/2$, whereas the true optimal is $0$. 

\section{Matroids and Pipage Rounding} \label{sec:pipage_rounding_appendix}
In this section, we provide the necessary background on matroids, as well as the lower tail versions of lemmas from~\cite{harvey2014pipage}, which let us prove Lemma~\ref{lem:renormalized-pipage}. 
\subsection{Matroids} \label{sec:matroids}
A pair $\mathcal{M} = (E,\mathcal{I})$ is a matroid if $E$ is a finite set and $\mathcal{I}$ is a collection of subsets of $E$ satisfying
\begin{enumerate}[(1)]
    \item If $I \in \mathcal{I}$ and $J \subseteq I$ then $J\in \mathcal{I}$, and
    \item If $I, J \in \mathcal{I}$ and $|I| < |J|$ then there is $e \in J\del I$ such that $I\cup\{e\} \in \mathcal{I}$.
\end{enumerate}
The sets in $\mathcal{I}$ are referred to as the \emph{independent sets} of the matroid $\mathcal{M}$. The maximal sets in $\mathcal{I}$ are called \emph{bases}, and it is a consequence of the matroid axioms that all bases have the same cardinality. For a subset $U \subseteq E$, we denote my $r(U)$ the maximum size of an independent set in $U$ and call this the rank of $U$. In this notation, we can say that every basis of $\mathcal{M}$ has cardinality exactly $r(E)$. Given a matroid $\mathcal{M}$, the matroid base polytope is the convex hull of indicator vectors of the bases of $\mathcal{M}$, and is denoted $\mathcal{P}(\mathcal{M})$. The base polytope has the following linear description
\begin{align*}
    \mathcal{P}(\mathcal{M})
    &= \text{conv}\left\{\chi(B) : B\text{ a basis of } \mathcal{M} \right\}\\
    &= \left\{ x \in \R^E : \sum_{e\in E}x_e = r(E), \sum_{e \in U} x_e \leq r(U) \,\,\forall U\subseteq E, x\geq 0\right\}.
\end{align*}
Cunningham ~\cite{Cunningham1984} showed that given $x \in \R^E_+$, it is possible to find a violated constraint for $\mathcal{P}(\mathcal{M})$ in strongly polynomial time using only an independence oracle for the matroid $\mathcal{M}$.

\subsection{Pipage Rounding}

The following theorem follows from the discussion in Section~\ref{sec:pipage_rounding}.

\begin{theorem}\cite{harvey2014pipage}\label{thm:pipage_jensen}
    There is a randomized polynomial time algorithm that, given $x_0 \in \mathcal{P}(\mathcal{M})$, outputs an extreme point $\hat x$ of $\mathcal{P}(\mathcal{M})$ with $\E[\hat x] = x_0$ and such that for any $g$ concave under swaps $\E[g(\hat x)] \leq g(x)$.
\end{theorem}

We will mainly make use of this theorem through the following claim. The conditions of the claim come from pessimistic estimators, but not all of them are strictly necessary. For a point $x \in [0,1]^n$, let $D(x)$ represent the corresponding product distribution over $\{0,1\}^n$ with marginals given by $x$.

\begin{lemma}\cite{harvey2014pipage}\label{lem:pipage_random_alg}
    Let $\mathcal{E} \subseteq \{0,1\}^n$ and $g: \mathcal{P}(\mathcal{M}) \rightarrow \R$ satisfy
    \begin{align*}
        &\Prob_{x\sim D(x)}[x \in \mathcal{E}] \leq g(x), \text{ and }\\
        &\min\{g(x - x_ie_i), g(x + (1-x_i)e_i)\} \leq g(x)
    \end{align*}
    for all $x \in [0,1]^n$, and $g$ be concave under swaps. If pipage rounding is started at an initial point $x_0 \in P$ and $\hat x$ is the random extreme point, then $\Prob[\hat x \in \mathcal{E}] \leq g(x_0)$.
\end{lemma}

Essentially, this lemma says that if we have a pessimistic estimator which is concave under swaps, then pipage rounding has the same type of concentration behavior as independent rounding, but will actually return a vertex of the matroid polytope.

For our particular application, we will be choosing the function $g$ to be an estimator for matrix concentration due to Tropp~\cite{tropp2015introduction}.

\begin{theorem}\cite[Theorem 5.1.1]{tropp2015introduction}\label{thm:tropp_concentration}
    Let $M_1,\ldots, M_n$ be self-adjoint matrices with $\lambda_{\max}(M_i) \leq R$ for all $i\in [n]$ and let $\mu_{\min} = \lambda_{\min}(\E_{x\sim \mathcal{D}(x)}[\sum_{i\in[n]} x_i M_i]$. For $t \in \mathbb{R}$, we have the bound
    $$ \Prob_{x \sim D(x)} \left[\lambda_{\min}\left( \sum_{i \in [n]} x_i M_i\right) \leq t\right] \leq \inf_{\theta < 0} g_{t, \theta}(x)$$ where $g_{t,\theta}(x) = e^{-\theta t} \cdot \mathrm{tr} \exp\left( \sum_{i\in[n]} \log \E e^{\theta x_i \cdot M_i}\right)$. Furthermore, for $t = (1-\epsilon) \mu_{\min}$, $$g_{t,\theta}(x) \leq d \cdot \left(\frac{e^{-\epsilon}}{(1-\epsilon)^{1-\epsilon}}\right)^{\mu_{\min}/R}.$$
\end{theorem}

This is the lower-tail version of the same concentration inequality which was used in~\cite{harvey2014pipage}. In that paper, they provide an upper-tail version of Lemma~\ref{lem:renormalized-pipage} using a new generalization of Lieb's concavity theorem, stated below.
    \begin{lemma}\cite{harvey2014pipage}\label{lem:lieb_concave}
        Let $L\in \mathbb{S}_d$, $C_1, C_2 \in \mathbb{S}_d^{++}$, and $K_1, K_2 \in \mathbb{S}_d^+$. Then the univariate function
        \begin{equation*}
            z \rightarrow \mathrm{tr} \exp \left( L + \log(C_1 + z K_1) + \log(C_2 - zK_2) \right)
        \end{equation*}
        is concave in a neighborhood of $0$.
    \end{lemma}

As a consequence, we get the following lemma.
\begin{lemma}\label{lem:matrix_swap_concavity}
    For $\theta < 0$, all $x \in [0,1]^m$, the function 
    \begin{equation*}
        g_{t,\theta}(x) =  e^{-\theta t} \cdot \mathrm{tr} \exp\left( \sum_{i\in[m]} \log \E_{x \sim \mathcal{D}(x)} e^{\theta x_i \cdot M_i}\right)
    \end{equation*}
    is concave under swaps.
\end{lemma}
\begin{proof}{(of Lemma~\ref{lem:matrix_swap_concavity})}
    Let $C_i := \E_{x \sim \mathcal{D}(x)}[ e^{\theta x_i \cdot M_i} ]= x_i\cdot e^{\theta M_i} + (1-x_i)\cdot I \succ 0$, and for any $i\in[n]$
    \begin{align*}
        \E_{x \sim \mathcal{D}(x+ze_i)}[ e^{\theta x_i \cdot M_i} ] = (x_i + z) e^{\theta M_i} + (1-x_i - z)\cdot I = C_i - z\cdot (I - e^{\theta M_i}).
    \end{align*}
    Then $\forall a, b \in [n]$, 
    \begin{align*}
        &g_{t, \theta}(x+z(e_a - e_b)) \\
        &= e^{-\theta t} \cdot \mathrm{tr}\exp \left( \sum_{i\in [n]\backslash\{a, b\}} \log C_i + \log \left(C_b + z\cdot(I - e^{\theta M_b}) \right) + \log \left(C_a - z\cdot (I - e^{\theta M_a})\right)  \right)\\
        &= e^{-\theta t} \cdot \mathrm{tr}\exp \left( L + \log \left(C_b + z\cdot K_b \right) + \log \left(C_a - z\cdot K_a\right) \right)\,,
        \end{align*}
    where $K_a = (I - e^{\theta M_a})$, $K_b = (I - e^{\theta M_b})$, and $L = \sum_{i\in [n]\backslash\{a, b\}} \log C_i \in \mathbb{S}_d$. 
    If $\theta \leq 0$, then $K_a\succeq 0$ and $K_b \succeq 0$. Using Lemma~\ref{lem:lieb_concave}, $z \rightarrow g_{t, \theta}(x+z(e_a - e_b))$ is concave in $z$, and the result follows.
\end{proof}
Combining Lemma~\ref{lem:matrix_swap_concavity} and Theorem~\ref{thm:tropp_concentration} with Lemma~\ref{lem:pipage_random_alg}, we obtain Lemma~\ref{lem:renormalized-pipage}.
\end{document}